\documentclass[12pt]{iopart}
\usepackage{multirow}
\expandafter\let\csname equation*\endcsname\relax
\expandafter\let\csname endequation*\endcsname\relax
\usepackage{amsmath}
\usepackage{amsthm}
\usepackage{cite}

\newtheorem{lemma}{Lemma}
\newtheorem{proposition}{Proposition}
\newtheorem{corollary}{Corollary}
\usepackage{amsfonts}
\newcommand{\id}{\text{id}}
\usepackage{braket}
\usepackage{caption}
\usepackage{xcolor}
\usepackage{subcaption}
\renewcommand\bra[1]{{\langle{#1}|}}
\renewcommand\ket[1]{{|{#1}\rangle}}
\newcommand{\op}[2]{|#1\rangle\langle #2|}
\newcommand{\mbb}[1]{\mathbb{#1}}
\newcommand{\mc}[1]{\mathcal{#1}}
\newcommand{\sgn}{\text{sgn}}
\usepackage{graphicx}
\usepackage{bm}
\usepackage{iopams}  
\usepackage[colorlinks=true,allcolors=blue]{hyperref}

\definecolor{cool_green}{rgb}{0.0, 0.5, 0.0}

\begin{document}

\title{Channel Activation of CHSH Nonlocality}
\author{Yujie Zhang$^{1,4}$, Rodrigo Araiza Bravo$^{1,2,4}$, Virginia O. Lorenz$^1$, Eric Chitambar$^{3,5}$}
\date{\today}
    \address{$^1$ Department of Physics, University of Illinois at Urbana-Champaign, Urbana, IL 61801, USA}
\address{$^2$ Department of Physics, Harvard University, Cambridge, Massachussets 02138, USA}
\address{$^3$ Department of Electrical and Computer Engineering,University of Illinois at Urbana-Champaign, Urbana, Illinois 61801, USA}
\address{$^4$ Authors contributed equally to this work}
\address{$^5$ Author to whom any correspondence should be addressed}
\ead{echitamb@illinois.edu}
\vspace{10pt}

\begin{abstract}
Quantum channels that break CHSH nonlocality on all input states are known as CHSH-breaking channels. In quantum networks, such channels are useless for distributing correlations that can violate the CHSH Inequality. Motivated by previous work on activation of nonlocality in quantum states, here we demonstrate an analogous activation of CHSH-breaking channels.  That is, we show that certain pairs of CHSH-breaking channels are no longer CHSH-breaking when used in combination.  We find that this type of activation can emerge in both uni-directional and bi-directional communication scenarios.
\end{abstract}

%
%
%
%
%

\section{Introduction}
The mystery of quantum mechanics involves many counter-intuitive phenomena absent in classical mechanics.  The most celebrated method for revealing the nonlocal features of quantum theory was proposed by John Bell in 1968, in what is now known as violating a Bell Inequality~\cite{2,3}. In recent years, nonlocality has been identified as a resource in quantum information theory~\cite{Chitambar-2019a}, with applications in quantum cryptography~\cite{5,6}, quantum key distribution~\cite{7} and quantum randomness~\cite{8}.
\par
In general, entanglement and nonlocality appear to be different resources~\cite{9}. It has been shown that quantum entanglement is required to generate nonlocal correlations~\cite{9a}, but entanglement is not sufficient for a quantum state to violate a Bell Inequality. That is, examples of entangled states admitting local hidden variable (LHV) models have been found~\cite{10}. In some cases, nonlocal behavior can still be exhibited after local measurement and post-selection~\cite{11,12,13}. This reveals that nonlocal correlations are subtle in form, and they can become manifest in different scenarios.
\par
In particular, nonlocal correlations are capable of being activated.  In general, activation means that two quantum objects can be combined to retrieve a particular quantum resource that was absent before the combination. Activation has been studied in the case of quantum channel capacities~\cite{14,15} and quantum entanglement~\cite{16,17}. Recently, this idea was also applied to quantum nonlocality. As shown by Navascu\'{e}s and V\'{e}rtesi~\cite{18}, two states $\rho_1$ and $\rho_2$, which cannot individually exhibit nonlocality in the so-called Clauser-Horne-Shimony-Holt (CHSH) scenario (i.e. two dichotomic observables per site), can nevertheless violate the CHSH Inequality when measured jointly.  A more general result was demonstrated by Palazuelos~\cite{19} in which two copies of a Bell local state $\rho$ can become nonlocal, a phenomenon known as ``super-activation''.  Other examples of nonlocality activation and super-activation can be found in ~\cite{20,21,22,23,24}.
\par
These previous works only considered activation of nonlocality on the level of quantum states.  This analysis can be understood from a resource-theoretic perspective in which nonlocality is regarded as a \textit{static} quantum resource, manifesting in different multipartite quantum states in different extents \cite{Chitambar-2019a}.  Activation then describes a particular way of harnessing this resource among two or more quantum states.  Alternatively, one could consider a \textit{dynamical} resource theory of quantum nonlocality in which nonlocality is a distinctive feature of dynamical quantum objects, i.e. quantum channels.  Quantum channels are of primary importance in many quantum information protocols such as quantum network communication~\cite{22}, quantum key distribution~\cite{8} and quantum teleportation~\cite{24a}.  While there are many ways a dynamical resource theory of nonlocality could be formulated \cite{Liu-2019a, Liu-2019b}, one approach involves identifying the nonlocality of a channel in terms of its ability to transmit nonlocal correlations. In its simplest form, a point-to-point quantum channel $\mathcal{E}^{A'\to B}$ distributes nonlocal correlations by sending one-half of an entangled state $\rho^{AA'}$ through the channel and locally measuring the joint state $\sigma^{AB}=\id^A\otimes\mathcal{E}^{A'\to B}(\rho^{AA'})$ where $\id^X$ is the identity map on subsystem $X$.  If the channel $\mathcal{E}^{A'\to B}$ is too noisy then $\sigma^{AB}$ will only be able to generate correlations that can be simulated by an LHV model.  If this holds for \textit{every} possible input state $\rho^{AA'}$, then $\mathcal{E}^{A'\to B}$ is a called a \textit{nonlocality-breaking} channel, as originally introduced by Pal and Ghosh~\cite{25}.  Such channels are analogous to the well-studied entanglement-breaking channels, which are those that break the entanglement between sender and receiver whenever they are used to distribute a quantum state~\cite{24b}.  While every entanglement-breaking channel is necessarily nonlocality-breaking, the results of~\cite{Almeida-2007a} imply that the converse is not true.

\par
In this paper, we focus on the family of CHSH-breaking channels $\mathcal{E}$.  These are channels whose output states $\id\otimes\mathcal{E}(\rho^{AA'})$ only generate local correlations when both parties choose between a pair of dichotomic observables. In section~\ref{sec2} we present the CHSH-breaking conditions for several channels. In section~\ref{sec3} we demonstrate activation by combining two CHSH-breaking channels and analyze this phenomenon as a nonlocality distribution task in two different situations. The results are summarized and discussed in section~\ref{sec4}. 

\section{Nonlocality and CHSH-breaking channels}
\label{sec2}

Consider Hilbert spaces $\mathcal{H}^{A}$ and $\mathcal{H}^B$ whose density matrices form the sets $D(\mathcal{H}^{A})$ and $D(\mathcal{H}^B)$. Mathematically, a quantum channel is a completely-positive, trace preserving map $\mathcal{E}:D(\mathcal{H}^{A})\rightarrow D(\mathcal{H}^{B})$ from $D(\mathcal{H}^{A})$ to $D(\mathcal{H}^{B})$. A quantum channel acting on a density matrix $\rho^A$ can be expressed as follows:
\begin{equation}\label{KrausRep}
\mathcal{E}(\rho^A)=\sum_{k}E_k \rho^A E_k^{\dagger},
\end{equation}
as a result of Choi’s theorem~\cite{26}, where the operators $\{E_k\}$ are known as Kraus operators and must satisfy the trace-preserving constraint $\sum_{k}E^\dagger_k E_k = \mathbb{I}$.

A channel $\mathcal{E}^{A'\to B}$ is called nonlocality-breaking if the output state $\sigma^{AB}=\id^A\otimes\mathcal{E}^{A'\to B}(\rho^{AA'})$ is Bell local for every input $\rho^{AA'}$ state; i.e. $\sigma^{AB}$ admits a LHV model for all local measurements.  This means that for any family of positive operator-valued measures (POVMs) $\{\Pi^A_{a|x}\}_{a,x}$ and $\{\Sigma^B_{b|y}\}_{b,y}$ on system $A$ and $B$, respectively, there exist conditional distributions $p(a|x,\lambda)$ and $p(b|y,\lambda)$ and shared variable $\lambda$ such that
\begin{equation}
    \tr[(\Pi^A_{a|x}\otimes\Sigma^B_{b|y})\sigma^{AB}]=\int d\lambda p(\lambda)p(a|x,\lambda)p(b|y,\lambda).
\end{equation}x
Since the distributions $p(ab|xy)$ admitting such a decomposition forms a compact, convex set, we can characterize the distributions by a collection of confining hyperplanes.  These hyperplanes correspond to so-called Bell inequalities, and their violation in the measurement statistics $p(ab|xy)=\tr[(\Pi^A_{a|x}\otimes\Sigma^B_{b|y})\sigma^{AB}]$ indicates that $\sigma^{AB}$ is not a Bell local state \cite{Peres-1999a}.  It is not difficult to show that $\mathcal{E}^{A'\to B}$ is nonlocality-breaking if and only if $\sigma^{AB}=\id^A\otimes\mathcal{E}^{A'\to B}(\op{\varphi}{\varphi}^{AA'})$ is Bell local for all pure-state inputs $\ket{\varphi}^{AA'}$ with system $A$ having dimension equaling $A'$ \cite{25}.
\par
Nonlocality-breaking channels are a generalization of the well-studied entanglement-breaking channels \cite{Horodecki-2003a}.  The latter refers to channels $\mathcal{E}^{A'\to B}$ such that $\sigma^{AB}=\id^A\otimes\mathcal{E}^{A'\to B}(\rho^{AA'})$ is separable for every input $\rho^{AA'}$.  Since every separable state is necessarily Bell local, it follows that every entanglement-breaking channel is nonlocality-breaking.  However, the converse is not true.  In~\cite{Almeida-2007a}, an LHV model was constructed for states having the form
\begin{equation}
    \sigma^{AB}=p\op{\varphi}{\varphi}^{AB}+(1-p)\rho^A_{\varphi}\otimes\frac{\mathbb{I}^B}{d_B},
\end{equation}
which corresponds to sending $\ket{\varphi}^{AA'}$ through the partially depolarizing channel $\mathcal{E}^{A'\to B}(X)=p X+\tr[X]\frac{1}{d_B}\mathbb{I}^B$.  It is known that this channel is entanglement-breaking whenever $p\leq \frac{1}{d_B+1}$ \cite{Horodecki-1999a}.  However, the local model of~\cite{Almeida-2007a} holds for values of $p>\frac{1}{d_B+1}$.  Thus, even on the level of channels, entanglement and nonlocality emerge as distinct quantum resources.

Given the complexity in deciding whether a given bipartite state is Bell local, in this paper we restrict attention to the CHSH Inequality, which is the only Bell inequality corresponding to the scenario of binary inputs and binary outputs \cite{30}.  Recall that for any set of observables $\{M_1^A,M_2^A,N_1^B,N_2^B\}$ with spectrum $\{-1,1\}$, the CHSH Inequality says that $|\Tr{(\mathcal{B}\rho^{AB})}|\le 2$, where
\begin{equation}
\mathcal{B}=M_1^A\otimes(N_1^B+N_2^B)+M_2^A\otimes(N_1^B-N_2^B)
\end{equation}
is called the Bell operator.  As in the general case, a channel $\mathcal{E}^{A'\to B}$ is called CHSH-breaking if $\sigma^{AB}=\id^A\otimes\mathcal{E}^{A'\to B}(\rho^{AA'})$ cannot violate the CHSH Inequality for any input state $\rho^{AA'}$, and we will refer to such states as being CHSH local.  Again, for deciding whether or not a channel is CHSH-breaking, it suffices to consider pure-state inputs.



To demonstrate activation phenomenon, it suffices to consider qubit channels, i.e. $\mathcal{H}^{A'}=\mathcal{H}^{B}=\mathbb{C}^{2}$.  States of a qubit system can be expressed as $\rho=\frac{1}{2}(\mathbb{I}+\boldsymbol{v}\cdot\boldsymbol{\sigma})$ where $\boldsymbol{v}\in\mathbb{R}^3$ is called the Bloch vector of the state and $\boldsymbol{\sigma}=(\sigma_x,\sigma_y,\sigma_z)$ is the Pauli vector.  In the two-level case, a channel can be characterized as an affine transformation $\mathbb{T}$ on the Bloch vector of the input state ~\cite{27,28}.  Explicitly, we have
\begin{equation}\label{channel}
\mathbb{T}=\begin{pmatrix}
             1 & \boldsymbol{0} \\
\boldsymbol{t} & \boldsymbol{\Lambda} \end{pmatrix},
\end{equation}
so that
\begin{equation}
    \mathcal{E}\left(\frac{1}{2}\left(\mathbb{I}+\boldsymbol{v}\cdot\boldsymbol{\sigma}\right)\right)=\frac{1}{2}\left(\mathbb{I}+(\boldsymbol{t}+\boldsymbol{\Lambda}\boldsymbol{v})\cdot\boldsymbol{\sigma}\right),
\end{equation}
where $\boldsymbol{t}$ is a real $1\times 3$ vector and $\boldsymbol{\Lambda}$ is a real $3\times 3$ matrix. Furthermore, $\boldsymbol{\Lambda}$ can be diagonalized under a proper unitary map on the input and output state $U_o\circ\mathcal{E}\circ U_i$~\cite{25}.  For the purposes of deciding whether or not a channel is CHSH-breaking, the unitaries $U_{i}$ and $U_o$ above can be absorbed by the state and measurements in the CHSH Inequality respectively.  Hence, the CHSH-breaking conditions only need to be derived for diagonal $\boldsymbol{\Lambda}$.
\par 
A special class of quantum channels called unital channels have the property that $\mathcal{E}(\id)=\id$. The following lemma says that for unital qubit channels, one only needs to consider a maximally entangled input to determine if it is CHSH-breaking.
\begin{lemma}
\label{UnitalThm}
A unital channel $\mathcal{E}$ is CHSH-breaking if and only if its output state $\id \otimes\mathcal{E}(\op{\Phi^{+}}{\Phi^{+}})$ does not violate the CHSH Inequality, where $\ket{\Phi^{+}}=\frac{1}{\sqrt{2}}(\ket{00}+\ket{11})$ is the maximally entangled state.
\end{lemma}
\par 
\noindent A proof of Lemma~\ref{UnitalThm} is provided in \ref{UnitalThmProof}.  For a unital channel represented by the affine transformation $\mbb{T}$ in the Bloch vector picture, it can be readily seen that $\boldsymbol{t}=\boldsymbol{0}$ and the diagonalized $\Lambda$ is parametrized by three real parameters $\{\lambda_1,\lambda_2,\lambda_3\}$.  It has been previously shown by Pal and Ghosh in~\cite{25} that $\id \otimes\mathcal{E}(\op{\Phi^{+}}{\Phi^{+}})$ cannot violate the CHSH Inequality for a unital channel $\mc{E}$ if and only if
\begin{equation}
\lambda_1^2+\lambda_2^2\le 1,\label{unital_break}
\end{equation}
assuming $|\lambda_1|\ge|\lambda_2|\ge|\lambda_3|$. Combining with Lemma~\ref{UnitalThm}, we thus conclude that Eq. ~(\ref{unital_break}) provides the CHSH-breaking condition for any unital qubit channel.
\par
For a general non-unital channel, we know of no analytical criteria for determining whether or not it is CHSH-breaking, and one typically obtains results by numerically searching over all input states and measurement settings~\cite{25}. Nevertheless, in~\ref{proof:amplitude} and~\ref{proof:erasure}, we provide analytical conditions for when special classes of non-unital channels are CHSH-breaking. 
\par
We now summarize our results of the CHSH-breaking conditions for certain families of qubit channels.
\par
\begin{itemize}
\item \textit{Depolarizing channel $\mathcal{E}_{d,p}$:} The depolarizing channel $\mathcal{E}_{d,p}$ perfectly transmits its input with probability $p$; with probability $1-p$ it throws away its input and outputs a completely mixed state.  On a two-qubit state, $\mathcal{E}_{d,p}$ acts as:
    \begin{equation}\label{DepolDescription}
      \id^A\otimes\mathcal{E}_{d,p}^{A'\to B}(\rho^{AA'})=p\rho^{AB}+ (1-p)\rho^A\otimes\frac{\mathbb{I}^B}{2}
    \end{equation}
    where $\rho^A=\Tr_{A'}(\rho^{AA'})$.  Kraus operators for this map are easily seen to be
    \begin{equation}\label{DepolOps}
    E_1=\sqrt{\frac{1+3p}{4}}\mathbb{I}, \quad E_i=\sqrt{\frac{1-p}{4}}\sigma_{i-1}\quad (2\leq i\leq 4).
    \end{equation}
     From Lemma \ref{UnitalThm}, it follows that $\mathcal{E}_{d,p}$ is CHSH-breaking if and only if $p\leq \frac{1}{\sqrt{2}}$.
\par
\item \textit{Amplitude damping channel  $\mathcal{E}_{a,p}$:} The amplitude damping channel $\mathcal{E}_{a,p}$ shrinks the $x$ and $y$ components of an input Bloch vector by a factor $\sqrt{p}$ while driving the $z$ component toward $+1$.  It has Kraus operator
    \begin{equation}\label{AmpDampOps}
   E_1=\op{0}{0}+\sqrt{p}\op{1}{1}, \quad E_2=\sqrt{1-p}\op{0}{1}.
    \end{equation}
    Note that this is a non-unital channel, and it is CHSH-breaking for $p\le\frac{1}{2}$, as calculated in \ref{proof:amplitude}.
\par
\item \textit{Loss channel $\mathcal{E}_{l,p}$:} The loss channel $\mathcal{E}_{l,p}$ perfectly transmits its input with probability $p$ whereas with probability $1-p$ it throws away its input and outputs the state $\ket{0}$.  On a two-qubit state, $\mathcal{E}_{l,p}$ acts as:
    \begin{equation}\label{LossDescription}
        \id^A\otimes\mathcal{E}_{l,p}^{A'\to B}(\rho^{AA'})=p\rho^{AB}+ (1-p)\rho^A\otimes\op{0}{0}^B,
    \end{equation}
and it has Kraus operators
   \begin{equation}\label{LossOps}
    E_1=\sqrt{p}\mathbb{I}, \quad E_2=\sqrt{\frac{1-p}{2}}\op{0}{1},\quad E_3=\sqrt{\frac{1-p}{2}}\op{1}{1}.
    \end{equation}
    This is a non-unital channel with a CHSH-breaking condition $p\le\frac{\sqrt{5}-1}{2}$, as shown in~\ref{proof:amplitude}.
\par
\item \textit{Erasure channel $\mathcal{E}_{e,p}$:}  The erasure channel $\mathcal{E}_{e,p}$ perfectly transmits its input with probability $p$; with probability $1-p$ it throws away its input and outputs a flag state $\ket{e}$, which is orthogonal to both $\ket{0}$ and $\ket{1}$.     On a two-qubit state, $\mathcal{E}_{e,p}$ acts as:
    \begin{equation}\label{erasureDescription}
       \id^A\otimes\mathcal{E}_{e,p}^{A'\to B}(\rho^{AA'})=p\rho^{AB}+ (1-p)\rho^A\otimes\op{e}{e}^B,
    \end{equation} 
and it has Kraus operators
    \begin{equation}\label{ErasureOps}
    E_1=\sqrt{p}\mathbb{I}, \quad E_2 = \sqrt{1-p}|e\rangle\langle 0|, \quad E_3=\sqrt{1-p}|e\rangle\langle 1|.
    \end{equation}
   The CHSH-breaking conditions is $p\le\frac{1}{2}$, as calculated in \ref{proof:erasure}.
\end{itemize}
\begin{figure}[th]
\centering
\subcaptionbox{\label{fig:OriProtDiag}}{\includegraphics[width=0.45\textwidth]{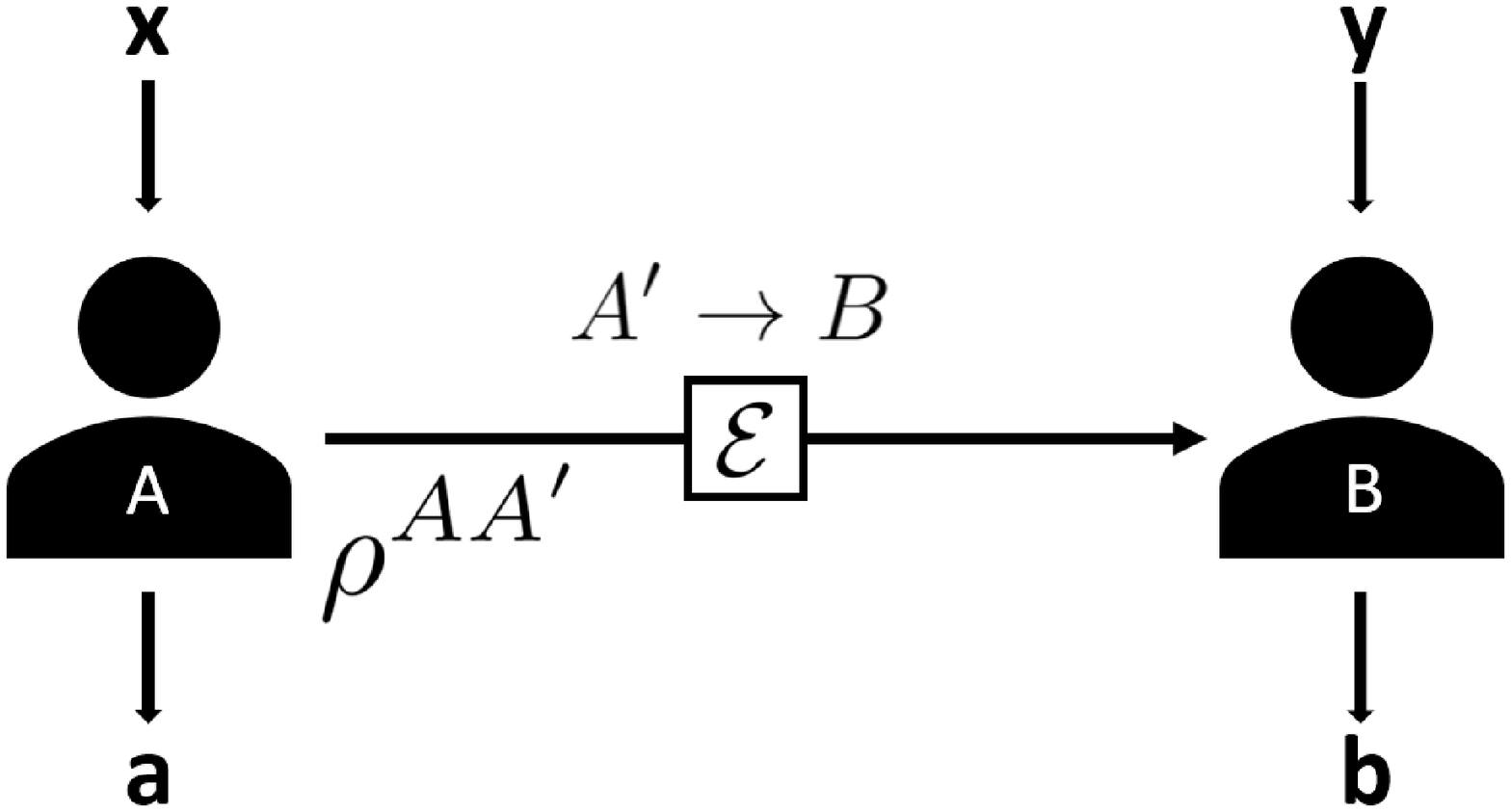}}%
\hfill

\subcaptionbox{\label{fig:UniProtDiag}}{\includegraphics[width=0.45\textwidth]{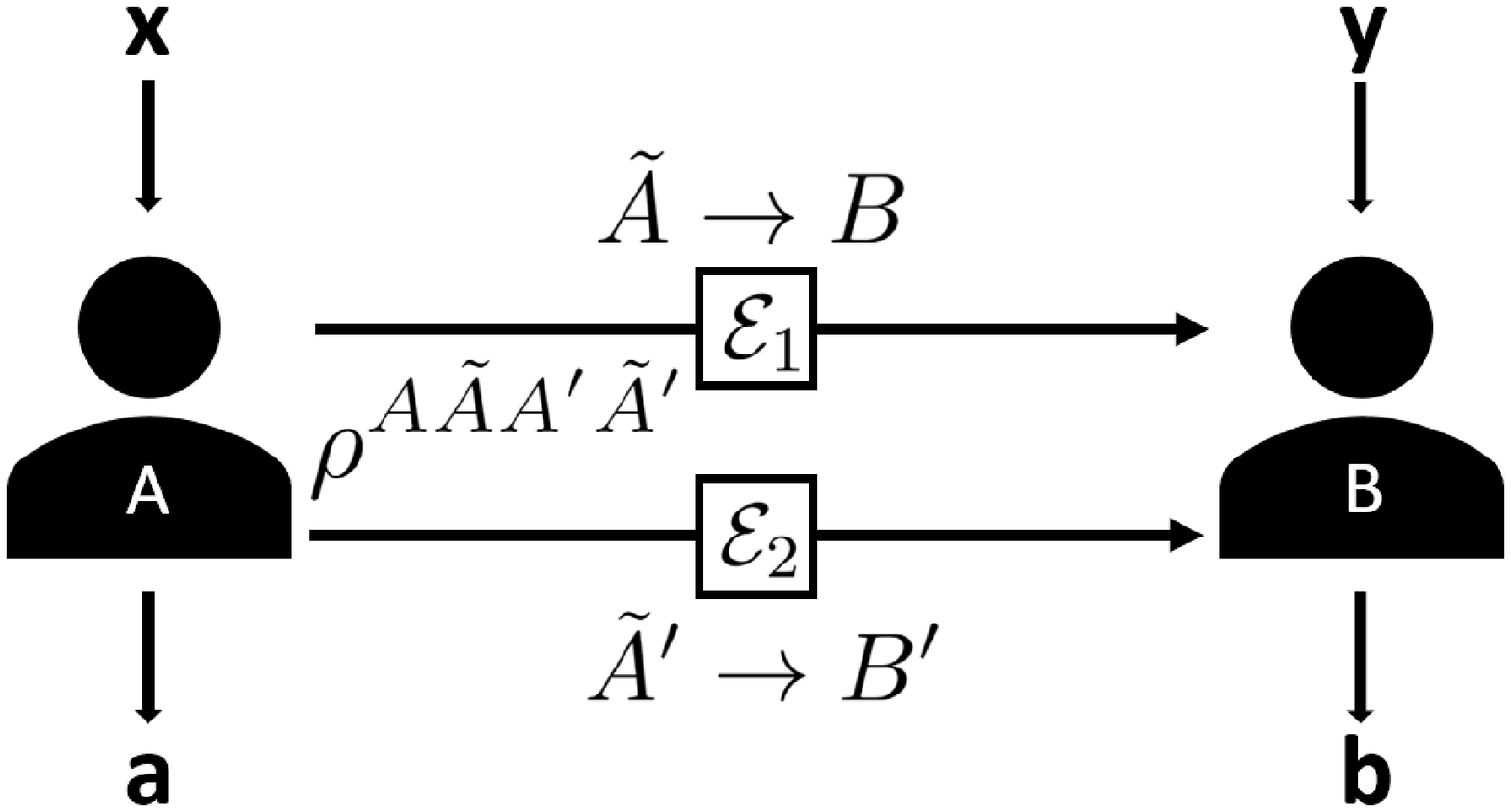}}%
\hfill
\subcaptionbox{\label{fig:BiProtDiag}}{\includegraphics[width=0.45\textwidth]{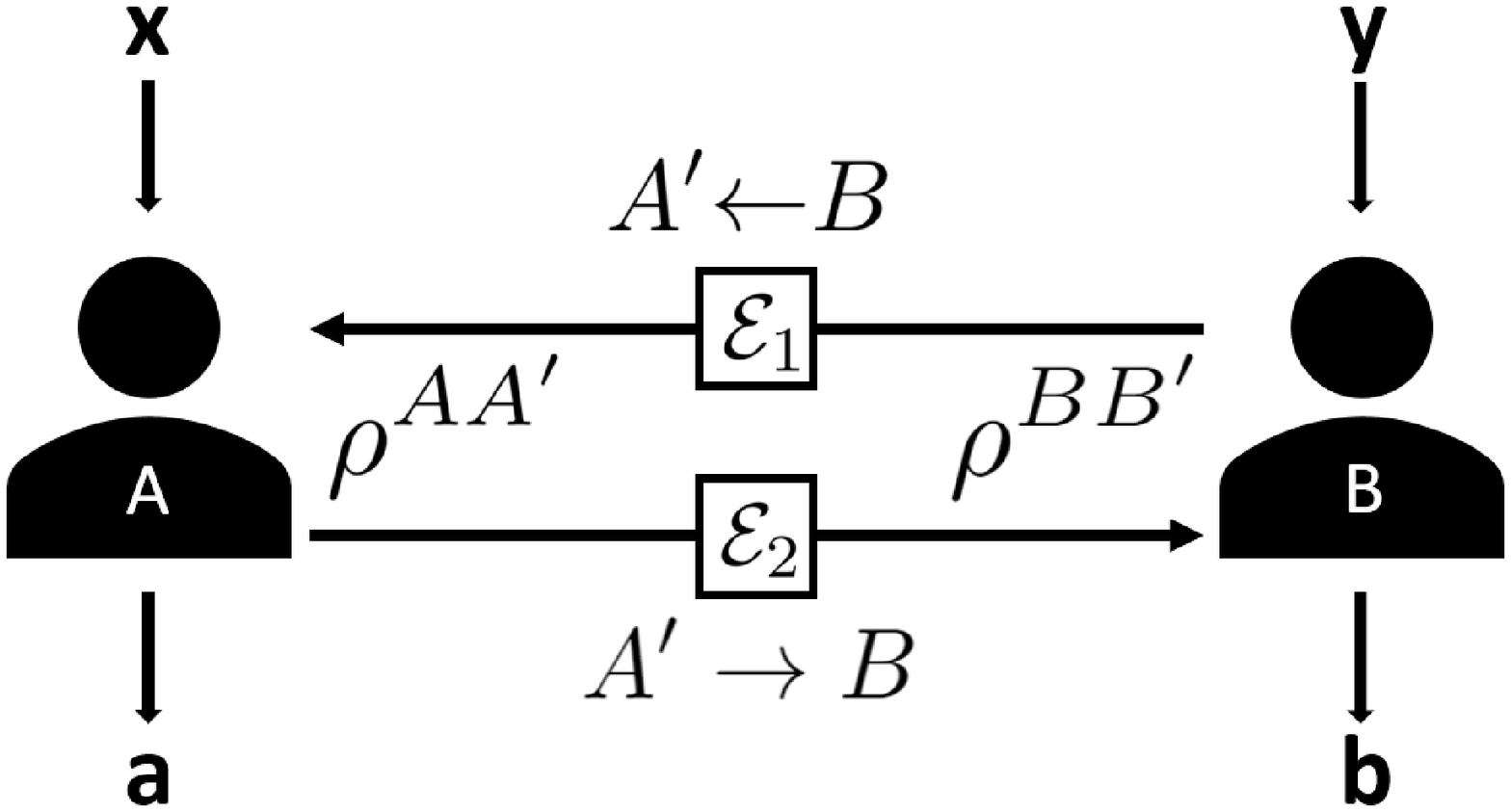}}%
\hfill
\par
\caption{(a) Original protocol: Alexis sends one particle to Bobby via a CHSH-breaking channel, no nonlocality can be retrieved under a two-setting-two-outcome scenario; (b) Unidirectional protocol: Alexis sends two particles to Bobby via two CHSH-breaking channels;  (c) Bidirectional protocol: Bobby sends one particle to Alexis, and Alexis sends one particle to Bobby via two CHSH-breaking channels.}
\end{figure}


\section{Channel Activation of CHSH nonlocality}
\label{sec3}
In \cite{18} it was shown that there exist two CHSH local states ${\rho}^{AB}$ and ${\rho}^{A^{'}B^{'}}$ such that ${\rho}^{AB}\otimes{\rho}^{A^{'}B^{'}}$ can violate the CHSH Inequality. We now show a similar nonlocality activation from the perspective of quantum channels.  Unlike quantum states, which are static resources, quantum channels are dynamical resources that can be used in different ways.  Two specific scenarios are discussed in this work. 
\par 
\noindent \textbf{Unidirectional protocol:}
Figure~\ref{fig:UniProtDiag} depicts a unidirectional way for using two quantum channels to distribute nonlocality between Alexis and Bobby. Alexis locally prepares a four-qubit state $\rho^{A\tilde{A}A'\tilde{A}'}$ and sends systems $\tilde{A}$ and $\tilde{A}'$ through two different qubit channels, $\mc{E}_1^{\tilde{A}\to B}$ and $\mc{E}_2^{\tilde{A}'\to B'}$.  This produces the state
\begin{equation}
\sigma^{ABA'B'}=\id\otimes\mathcal{E}^{\tilde{A}\to B}_1\otimes\id\otimes\mathcal{E}^{\tilde{A}'\to B'}_2(\rho^{A\tilde{A}A'\tilde{A}'}).
\label{4-state}
\end{equation}
CHSH activation is achieved if this state can violate the CHSH Inequality along the $AA':BB'$ cut when $\mc{E}_1$ and $\mc{E}_2$ are both CHSH-breaking channels.
\par
\noindent \textbf{Bidirectional protocol:}
Figure~\ref{fig:BiProtDiag} shows a bidirectional method for distributing nonlocality between Alexis and Bobby. Unlike a unidirectional protocol, both Alexis and Bobby locally prepare a two-qubit state, and they each send one of their qubits to the other party through a channel. The final shared state will have the form
\begin{equation}
\sigma^{AB}\otimes\sigma^{A'B'}=\id\otimes\mathcal{E}^{{A}'\to B}_1\otimes\mathcal{E}^{B\to A'}_2\otimes\id(\rho^{AA'}\otimes\rho^{BB'}).
\label{2-state}
\end{equation}
Again, CHSH activation is achieved if this state can violate the CHSH Inequality along the $AA':BB'$ cut when $\mc{E}_1$ and $\mc{E}_2$ are both CHSH-breaking channels.
\par
In both of these scenarios,  CHSH nonlocality of the output state can be detected using the Bell operator:
\begin{equation}
\mathcal{B}=M_1^{AA'}\otimes(N^{BB'}_1+N^{BB'}_2)+M_2^{AA'}\otimes (N^{BB'}_1-N^{BB'}_2),
\label{Bell2}
\end{equation}
where $\{M^{AA'}_1,M^{AA'}_2,N^{BB'}_1,N^{BB'}_2\}$ are joint observables for party $AA'$ and party $BB'$.  The observables have spectrum $\{-1,1\}$, and the CHSH Inequality says ${\Tr{(\mathcal{B}\sigma^{AA':BB'})}\le 2}$.  We next describe our general approach for optimizing the value of $\Tr{(\mathcal{B}\sigma^{AA':BB'})}$.

\subsection{The See-Saw Optimization Algorithm}
The maximum CHSH-value for given channels $\mathcal{E}_{1}\otimes\mc{E}_2$ can be obtained by maximizing observables $\{M_x,N_y\}$ and input states $\rho$ using algorithms like the ones in~\cite{18,31}.   For example, consider a unidirectional protocol using channels $\mc{E}_1\otimes\mc{E}_2$.  For a fixed input state $\rho$ and observable choices $N_y$ on Bobby's side, one can define
\begin{equation*}
\begin{split}
F_{1}&=\Tr_{BB'}([\mathbb{I}^{AA'}\otimes(N_1+ N_2)]\sigma^{ABA'B'}),\\
F_{2}&=\Tr_{BB'}([\mathbb{I}^{AA'}\otimes(N_1- N_2)]\sigma^{ABA'B'}),
\end{split}
\end{equation*}
where $\sigma^{ABA'B'}$ is given by Eq.~(\ref{4-state}). As shown in~\cite{18}, the optimal observables $M_x$ for Alexis are given by $M_x=\sum_{i}\sgn(\lambda_{x,i})\op{e_{x,i}}{e_{x,i}}$ where $F_x=\sum_i\lambda_{x,i}\op{e_{x,i}}{e_{x,i}}$.  Optimal observables $N_y$ can be likewise obtained for a fixed input state $\rho$ and observables $M_x$ on Alexis's side.  Finally, for fixed local observables $\{M_x,N_y\}$ (and hence fixed Bell operator $\mc{B}$), the optimal input state can be determined by observing that
\begin{align}
    \tr(\mc{B}\sigma^{AA':BB'})&=\tr\left[\mc{B}(\mathcal{E}^{\tilde{A}\to B}_{1}\otimes\mc{E}^{\tilde{A}'\to B'}_2)[\rho^{A\tilde{A}A'\tilde{A}'}]\right]\notag\\
    &=\tr\left[\rho^{A\tilde{A}A'\tilde{A}'}(\mathcal{E}^{\dagger}_{1}\otimes\mc{E}^{\dagger}_2)[\mc{B}]\right],
\end{align}
where $\mc{E}_i^\dagger$ denotes the adjoint CP map of $\mc{E}_i$.  Thus we can maximize $\tr(\mc{B}\sigma^{AA':BB'})$ by choosing the input state to be the outer product of the eigenstate associated with the largest eigenvalue of $\mathcal{E}^{\dagger}_{1}\otimes\mc{E}^{\dagger}_2[\mc{B}]$.  

We have just observed how for any two of the three problem variables $\{M_x,N_y,\rho\}$, we can always choose the third so as to maximize the value of $\tr(\mc{B}\sigma^{AA':BB'})$.  This suggests we perform a ``see-saw'' algorithm in which we alternate optimizing $\tr(\mc{B}\sigma^{AA':BB'})$ over one of the three variables \cite{18}.  For channels $\mc{E}_1\otimes\mc{E}_2$, the algorithm is then as follows:
\begin{enumerate}
    \item Randomly initialize $M^x$, $N^y$, and $\rho$.
    \item Update each of the variables $\{M_x,N_y,\rho\}$ in the order $M^x\rightarrow\rho\rightarrow N^y\rightarrow\rho$ repeatedly by applying the optimization procedures described above. 
    \item In step 2, it is possible that after reaching a local maxima in the CHSH-value, the algorithm isn't able to improve on the CHSH-value. When the same CHSH-value is reached, say, 10 times after consecutive updates of $M_x,N_y$, and $\rho$, we suddenly change the state to $\rho\rightarrow (1-\varepsilon)\rho+\varepsilon\rho^*$ with certain probability $\varepsilon$ where $\rho^*$ is  some previously fixed state. Meanwhile, we also store the maxima that was reached before the sudden change.
    \item Steps 2 and 3 are repeated until the same maximal is reached repeatedly. 
\end{enumerate}

Although we tried different states $\rho^{*}$ for each channel, we highlight that using a random entangled pure state sufficed. To do this, we picked $\rho^{*} = U^{\dagger}\ket{\lambda}\bra{\lambda}U$ where $U$ and $\lambda\in (0,1)$ are a random unitary and real parameter, and $\ket{\lambda} = \sqrt{\lambda}\ket{10}+\sqrt{1-\lambda}\ket{10}$.

\subsection{Results}

For both protocols, we are able to see activation for several combinations of quantum channels.  Some of the activation results can be found in Table~\ref{ActResults}, and the full numerical code we used can be found at \cite{Code}.  Our largest violation is obtained for two amplitude damping channels in a bidirectional protocol, in which we obtain a CHSH value of $2.011\,91$.  
\begin{table}[h]
 \caption{Activation results for both protocols. All of these violations are calculated in the case where the channels are CHSH-breaking.}
\footnotesize
\label{ActResults}
\centering
\begin{tabular}{llll}
\br
\textrm{Protocol}&
\textrm{Channel 1}&
\textrm{Channel 2}&
\textrm{Maximum violation}\\
\mr
\multirow{2}{*}{{Unidirectional}} & $\mathcal{E}_{a,p=1/2}$ & $\mathcal{E}_{d,p=1/\sqrt{2}} $  & $2.005\,41$     \\ 
                            &  $\mathcal{E}_{e,p=1/2}$ & $\mathcal{E}_{d,p=1/\sqrt{2}} $      & $2.004\,84$      \\ 
                            & & \\
\multirow{4}{*}{{Bidirectional}} &  $\mathcal{E}_{a,p=1/2}$ & $\mathcal{E}_{a,p=1/2} $         & $2.011\,91$     \\ 
                                            &  $\mathcal{E}_{e,p=1/2}$ & $\mathcal{E}_{e,p=1/2} $              & $2.001\,64$\\ 
                                            & $\mathcal{E}_{a,p=1/2}$ & $\mathcal{E}_{l,p=(\sqrt{5}-1)/2} $ &$2.002\,11$ \\
                                            & $\mathcal{E}_{a,p=1/2}$ & $\mathcal{E}_{l,p=1/2} $ &$2.000\,31$ \\
\br
\end{tabular}
\end{table}
Notice that every instance of channel activation in a bidirectional protocol provides an instance of state activation.  Indeed, if $\mc{E}_1^{A'\to B}$ and $\mc{E}_2^{B\to A'}$ are CHSH-breaking states such that $\id\otimes\mc{E}_1^{A'\to B}\otimes\mc{E}_2^{B\to A'}\otimes\id(\rho^{AA'}\otimes\rho^{BB'})$ violates the CHSH Inequality, then the states $\sigma_1^{AB}=\id\otimes\mathcal{E}^{A'\to B}_1(\rho^{AA'})$ and $\sigma_2^{A'B'}=\mathcal{E}^{B\to A'}_2\otimes\id(\rho^{BB'})$ are both CHSH local, but they can be activated when put together.  

From the aforementioned result of the amplitude damping channel in the bidirectional protocol, we thus have a state CHSH activation of $2.011\,91$.  In comparison, the maximal activation value found by Navascu\'{e}s and V\'{e}rtesi~\cite{18} is $2.023\,24$, and it was obtained by considering a subset of CHSH-local states that are 2-extendable. In our analysis, we do not have a greater violation, yet our method has led to finding state activation results for some symmetric states which were not found in~\cite{18}.  Specifically, for two CHSH-breaking amplitude damping channels, we have found a violation of $2.011\,72$ using CHSH-local states $\sigma_1^{AB}$ and $\sigma_2^{AB}$ satisfying the condition $\sigma_1^{AB}=\mbb{F}_{AB}\sigma_2^{AB}\mbb{F}_{AB}$, where $\mbb{F}_{AB}$ is the SWAP operator between Alexis and Bobby's systems.  This violation is only slightly smaller than our maximal violation, and it can be very useful for demonstrating super-activation of a symmetric CHSH-local state.  Namely, by introducing ancillary qubits, we define a new state of the form:
\begin{equation}
\tilde{\sigma}^{abAB}=\frac{1}{2}(\op{0}{0}^a\otimes\op{1}{1}^b\otimes{\sigma}_1^{AB}+\op{1}{1}^a\otimes\op{0}{0}^b\otimes \sigma_2^{AB}),
\end{equation}
By construction, this state is $aA\leftrightarrow bB$ symmetric, and it does not violate the CHSH Inequality since $\sigma_1$ and $\sigma_2$ are CHSH-local. On the other hand, a CHSH super-activation can be easily shown (i.e. $(\tilde{\sigma}^{abAB})^{\otimes 2}$ violates the CHSH Inequality) according to the following scheme: Alexis and Bobby both perform measurements on their ancilla qubits in the $\{\ket{0},\ket{1}\}$ basis;
\begin{itemize}
    \item If both ancilla qubits are in $\ket{0}$ or $\ket{1}$, Alexis (Bobby) measures on the $A_1A_2$~($B_1B_2$) subsystem with $\mathbb{I}^{A_1A_2}$ ($\mathbb{I}^{B_1B_2}$);
    \item If the measurement result is $\ket{0}^{a_1}\otimes\ket{1}^{a_2}$ ($\ket{1}^{b_1}\otimes\ket{0}^{b_2}$), Alexis (Bobby) measures on the $A_1A_2$~($B_1B_2)$ subsystem with $M_x^{A_1A_2}$ ($N_y^{B_2B_2}$);
    \item If the measurement result is $\ket{1}^{a_1}\otimes\ket{0}^{a_2}$  ($\ket{0}^{a_1}\otimes\ket{1}^{b_2})$, Alexis (Bobby) measures on the $A_1A_2$~($B_1B_2)$ subsystem with $\mbb{F}_{A_1A_2}M_x^{A_1A_2}\mbb{F}_{A_1A_2}$ ($\mbb{F}_{B_1B_2}N_y^{B_1B_2}\mbb{F}_{B_1B_2}$),
\end{itemize}
where $\{M_x,N_y\}$ is the optimal measurement strategy for the state $\sigma^{A_1B_1}_1\otimes\sigma^{A_2B_2}_2$.  Clearly, a CHSH violation of $(2\times2.011\,72+2\times2)/4=2.005\,86$ can be obtained. Hence, we have a four-qubit state $\tilde{\sigma}^{abAB}$ that is CHSH-local and symmetric, yet it can be super-activated.  By comparison, this scheme uses fewer ancillary qubits than the analogous example presented in~\cite{18}.
\par
As another remark, we have obtained a fairly robust activation result with a CHSH violation of $2.000\,31$ for the loss channel at $p=1/2$; however the channel becomes CHSH breaking for $p<(\sqrt{5}-1)/2~$. It is interesting to place this result alongside the findings of Ref.~\cite{hidden}, where the state  $\sigma^{AB}=\frac{1}{2}\op{\Psi^{-}}{\Psi^{-}}+\frac{1}{4}{\mathbb{I}}\otimes\op{0}{0}$ has been shown to be Bell local under all projective measurements. That is, the loss channel $\mathcal{E}_{l,p}$ cannot violate \textit{any} Bell Inequality at $p=1/2$ using projective measurements when the maximally entangled state $\op{\Psi^{-}}{\Psi^{-}}$ is the input state. Whether the channel is nonlocality breaking at $p=1/2$ for all inputs is unknown; yet, if this were true, the CHSH violation here would demonstrate channel activation of general nonlocality (rather than just CHSH nonlocality) under projective measurements. We leave this as an important open problem to resolve.
\begin{figure}[h]
\centering
\includegraphics[width=0.7\textwidth]{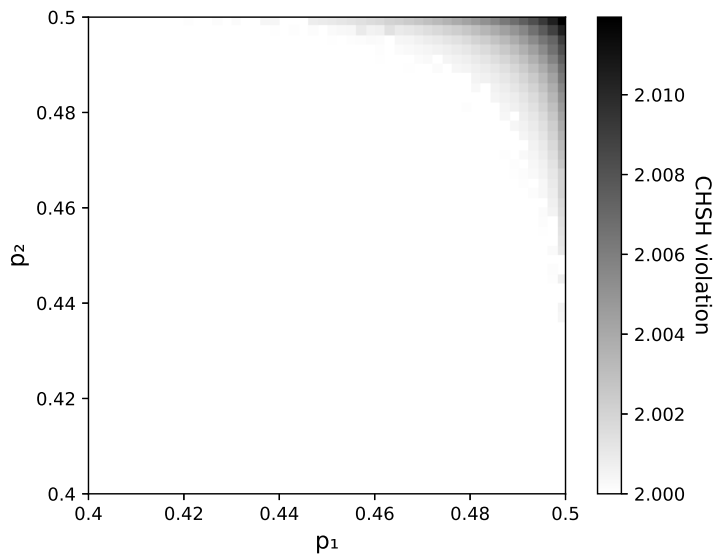}
\caption{Channel activation of CHSH nonlocality with two amplitude damping channels $\mathcal{E}_{a,p_1}$ and $\mathcal{E}_{a,p_2}$. The maximal CHSH violation of $2.011\,91$ is obtained at $p_1=p_2=0.5$, where both channels just become CHSH-breaking.}
\label{fig:robustness}
\end{figure}
\subsection{Robustness of Activation}
Most of our activation results can be obtained even when the channel parameters are larger than the critical values at which they become CHSH-breaking. This highlights a certain robustness to noise in our activation results, which is crucial for experimentally demonstrating activation of CHSH-breaking channels.  Additionally, the mapping of the CHSH-activation region shows that the see-saw algorithm indeed converges for a large number of points. 
\par
To illustrate this activation robustness to noise, we take the case of two amplitude damping channels in the bidirectional scenario as an example. By tuning the noise parameter $p_{1(2)}$ for the two amplitude damping channels, we calculate the CHSH violation in the region $p_{1(2)}\in[0.4,0.5]$ with an interval of 0.002, as shown in Figure \ref{fig:robustness}. Since the amplitude damping channel is CHSH-breaking when $p\le 0.5$, the CHSH violation here signifies a valid CHSH activation, and it is seen to persist in a small region around $p_{1(2)}=0.5$. This in turn, confirms and supports our activation result in Table \ref{ActResults}. However, when adding more noise, there is no guarantee that our algorithm will continue to converge. As can be seen from Figure \ref{fig:robustness}, some points indicating no activation are surrounded by points of activation.  By a convexity argument, these points of ``no activation'' can indeed be activated; however, the figure merely indicates that the see-saw algorithm failed to converge there.   
\section{Conclusion}
\label{sec4}
In this paper, we have investigated different qubit channels that prohibit the distribution of CHSH-violating quantum states.  On their own, such channels are useless for any quantum information task that involves CHSH nonlocality.  We provide simple criteria for determining the CHSH-breaking condition for all unital channels, and we perform analytical calculations that determine the CHSH-breaking conditions for some special channels that naturally arise in scenarios with experimental noise.  Our main result is that two CHSH-breaking channels may no longer be CHSH-breaking when used in parallel.  We demonstrate this result in two different scenarios.  Interestingly, in the bidirectional scenario, our activation and super-activation results do not use states that are entangled across both input systems.  This is in contrast to other channel activation results in which the input states are entangled.

 The activation findings presented here show that certain channels can be used to distribute nonlocality over long distances only when used in parallel with other channels. These results are particularly useful in instances of long-distance and noisy quantum communication.  However, it remains an open question as to whether there exist complete nonlocality-breaking channels that can nevertheless be activated.  This requires producing local models for all channel outputs and testing the breakage of all Bell Inequalities.   
\appendix

\ack
We thank Kai Shinbrough, Xinan Chen, Nicholas Laracuente, Marius Junge, and Graeme Smith for useful discussion. This work was supported by the National Science Foundation Award No. 1839177.
\appendix
\section{Proof of Lemma~\ref{UnitalThm}: CHSH-Breaking Condition for all Unital Channels}\label{UnitalThmProof}
For two qubits, the CHSH operator can be written as a correlation operator, which is some linear operator lying in the linear span of $\{\sigma_i\otimes\sigma_j\}_{i,j=1}^3$.  Explicitly, we have
\begin{equation}
    \mathcal{B}=\boldsymbol{a}_1\cdot\boldsymbol{\sigma}\otimes(\boldsymbol{b}_1+\boldsymbol{b}_2)\cdot\boldsymbol{\sigma}+\boldsymbol{a}_2\cdot\boldsymbol{\sigma}\otimes(\boldsymbol{b}_1-\boldsymbol{b}_2)\cdot\boldsymbol{\sigma},
\end{equation}
where $\boldsymbol{a}_i,\boldsymbol{b}_j$ are spin directions for Alexis and Bobby's measurements and the  corresponding correlation matrix is given as $B=\boldsymbol{a}_1\otimes(\boldsymbol{b}_1+\boldsymbol{b}_2)+\boldsymbol{a}_2\otimes(\boldsymbol{b}_1-\boldsymbol{b}_2)$, with rank not greater than two, which indicates that we can always diagonalize it as:
\begin{equation}
 \mathcal{B}=d_1\sigma_1\otimes\sigma_1+d_3\sigma_3\otimes\sigma _3.
\end{equation}
The Proof of Theorem 1 then follows from the following proposition.
\par
\begin{proposition}
\label{Prop:full-correlation-max}
Let $T=\sum_{i,j=1}^3t_{i,j}\sigma_i\otimes \sigma_j$ be any  correlation operator having correlation matrix $t_{ij}$ with rank less than 3. Then the expectation value {\upshape $\Tr(T\rho)$} is maximized by a maximally entangled state. 
\end{proposition}
\begin{proof}
Clearly the expectation value is maximized by some pure state $\ket{\psi}$. It then suffices to prove the proposition for diagonal correlation operators, $T=d_1\sigma_1\otimes\sigma_1+d_3\sigma_3\otimes\sigma_3$, since any Bell correlation operator can be converted into this form by local unitaries, and the latter does not change the entanglement of the maximizing state.  Write an arbitrary pure state as $\ket{\psi}=R\otimes \mathbb{I}\ket{\Phi^+}$, where $\ket{\Phi^{+}}=(\ket{00}+\ket{11})\sqrt{2}$ and $R$ is some matrix satisfying $\psi^A:=\Tr_{A'}(\op{\psi}{\psi})=RR^\dagger$ with $\Tr{(\psi^A)}=1$.  Then an application of the Cauchy-Schwartz Inequality gives
\begin{align}
\bra{\psi}T\ket{\psi}&=d_1\Tr[R^\dagger\sigma_1 R \sigma^{T}_1]+d_3\Tr[R^\dagger\sigma_3 R \sigma^{T}_3]\\
&\leq (d_1+d_3)\Tr[R^\dagger R]=d_1+d_3.
\end{align}
This upper bound is attained by taking $R=\mathbb{I}/2$, which corresponds to a maximally entangled state.
\end{proof}
\begin{corollary}
If $\mathcal{E}$ is a unital channel and $\mathcal{B}$ is any CHSH operator, the CHSH value $\Tr[\mathcal{B}\left(\id\otimes\mathcal{E}(\op{\psi}{\psi})\right)]$ is maximized by a maximally entangled state.
\end{corollary}
\begin{proof}
We have $\Tr[\mathcal{B}\left(\id\otimes\mathcal{E}(\op{\Phi^{+}}{\Phi^{+}})\right)]$ which is equivalent to $\Tr[\op{\Phi^{+}}{\Phi^{+}}\left(\id\otimes\mathcal{E}^{\dagger}(\mathcal{B})\right)]$ where $\mathcal{E}^\dagger$ is the dual map of $\mathcal{E}$.  If $\mathcal{E}$ is unital, then $\Tr[\mathcal{E}^\dagger(\sigma_i)]=0$, which means that $T=\id\otimes\mathcal{E}^{\dagger}(\mathcal{B})$ is still a correlation operator having correlation matrix with rank less than 3.  From the previous proposition, the corollary follows.
\end{proof}
\section{CHSH-breaking Conditions for the Amplitude Damping and Loss Channels}\label{proof:amplitude}

As we have shown in the previous section, for a unital channel to be CHSH-breaking it is sufficient to show that the channel is CHSH-breaking for the maximally entangled state. However, this is usually incorrect for general quantum channels; some counterexamples can be found in~\cite{25}. \par
For a non-unital channel, there is no such simple test. Here we analytically find the CHSH-breaking condition for a specific class of channels where
\begin{equation}
\mathbb{T}=\begin{pmatrix}1&0&0&0\\
0&\lambda_1&0&0\\
0&0&\lambda_1&0\\
t_3&0&0&\lambda_3
\end{pmatrix}.
\end{equation}
As can be checked, both the amplitude damping channel and the loss channel belong to this class. 
\par
For a pure state $\ket{\tilde{\psi}}=(UW_{\lambda}V^T\otimes\mathbb{I})\ket{\Phi^{+}}$ where $W_{\lambda}=\sqrt{\lambda}\op{0}{0}+\sqrt{1-\lambda}\op{1}{1}$ with $\lambda\ge 1/2$, the correlation function can be expressed as~\cite{25}:
\begin{equation}
T=\begin{pmatrix}\alpha R_{11}\lambda_1&-\alpha R_{21}\lambda_1&\alpha R_{31}\lambda_3\\
\alpha R_{12}\lambda_1&-\alpha R_{22}\lambda_1&\alpha R_{32}\lambda_3\\
R_{13}\lambda_1&-R_{23}\lambda_1&R_{33}\lambda_3+\sqrt{1-\alpha^2}t_3\end{pmatrix},
\end{equation} 
where $\alpha=2\sqrt{\lambda(1-\lambda)}$ and $R\in SO(3)$ is a real rotation matrix corresponding to $V^T\in SU(2)$; hence, by using the orthogonality $\sum_j R_{ij}R_{kj}=\delta_{ik}$, we have $H=TT^{\dagger}$:
\begin{equation}
H=\alpha^2\lambda_1^2\begin{pmatrix}1&0&0\\
0&1&0\\
0&0&1\end{pmatrix}+H',\label{APbreaking}
\end{equation}
where $H'$ is not full rank and not positive semi-definite when $\lambda_3^2\le\lambda_1^2$. Hence, it will have at most one positive eigenvalue. 
\par
Based on the Horodecki criteria~\cite{32}, the CHSH-breaking condition is determined by $\lambda_i+\lambda_j\le1$ where $\lambda_i,\lambda_j(i\ne j)$ are the two largest eigenvalues of $H$. Therefore, from~(\ref{APbreaking}), the CHSH-breaking condition is given by $2\alpha^2\lambda_1^2+\lambda_{l}\le 1$, where $\lambda_{l}$ is the largest eigenvalue of $H'$ over all choices of $R$.
\par For the case of the amplitude damping channel with $\lambda_1=\lambda_2=\sqrt{p},\lambda_3=p, t_3=1-p$, 
\begin{equation}
H'=-p(1-p)\boldsymbol{a}^T\boldsymbol{a}+(1-\alpha^2)\boldsymbol{b}^T\boldsymbol{b},
\end{equation}
where $\boldsymbol{a}=\{\alpha R_{31},\alpha R_{31},R_{33}-(1-\alpha^2)\}^T$, and $ \boldsymbol{b}=\{0,0,1\}^T$. The largest eigenvalue of $H'$ is $\lambda_l=1-\alpha^2$ obtained when $\boldsymbol{a}\perp\boldsymbol{b}$. Hence the CHSH-breaking condition will be given by
\begin{equation}
2\alpha^2p+1-\alpha^2=\alpha^2(2p-1)+1\le1.
\end{equation}
As a result, $p\le \frac{1}{2}$ is the CHSH-breaking condition for the amplitude damping channel. 
\par
Similarly, for the loss channel with $\lambda_1=\lambda_2=\lambda_3=p, t_3=1-p$ , $H'$ is given by:
\begin{equation}
\sqrt{1-\alpha^2}\begin{pmatrix}
0&0&\alpha p(1-p)R_{31}\\
0&0&\alpha p(1-p)R_{32}\\
\alpha p(1-p)R_{31}&\alpha p(1-p)R_{32}&2p(1-p)R_{33}+\sqrt{1-\alpha^2}[p^2+(1-p)^2]
\end{pmatrix},
\end{equation}
the largest eigenvalue $\lambda_l=2\sqrt{1-\alpha^2}p(1-p)+({1-\alpha^2})[p^2+(1-p)^2]$ is achieved when $R_{33}=1$. We have the CHSH-breaking condition
\begin{equation}
1+2({1-\sqrt{1-\alpha^2}})\left[(1-p)^2-(2+\sqrt{1-\alpha^2})(1-p)+\frac{\sqrt{1-\alpha^2}+1}{2}\right]\le 1,
\end{equation}
which gives $p\le\frac{[{(\sqrt{1-\alpha^2}+1)^2+1}]^{1/2}-\sqrt{1-\alpha^2}}{2}$; therefore, the CHSH-breaking condition for the loss channel is $p\le \frac{\sqrt{5}-1}{2}$, which is attained when $\alpha\rightarrow 0$.
\section{CHSH-breaking Conditions for the Erasure Channel}\label{proof:erasure}

Since the erasure channel in~(\ref{erasureDescription}) maps a qubit system to a qutrit system, the Horodecki criterion~\cite{32} cannot be directly applied. However, since the channel is basis-independent, we only need to check the CHSH-breaking condition of state $\ket{\tilde{\psi}}=\lambda\ket{00}+\sqrt{1-\lambda}\ket{11}$ with $\lambda\ge 1/2$.\par
Consider the measurement $M^{A}_1,M^{A}_2,N^{B}_1,N^{B}_1$ with spectrum $\{1,-1\}$, where $M^{A}_{1(2)}$ act on the qutrit system and $N^{B}_{1(2)}$ act on the qubit system. We can assume $M^{A}_{1(2)}$ to be of the following form:
\begin{equation}
\begin{pmatrix}
m_{11} &m_{12}&0\\
m_{21} &m_{22}&0\\
0&0&\pm 1
\end{pmatrix}.
\end{equation}
Without loss of generality, we take $m_{33}=1$ for both $M^{A}_1$ and $M^{A}_2$. This gives us
\begin{equation}\label{twoTraces}
\Tr{(\mathcal{B}\rho^{AB})}=p\Tr{(\mathcal{B}'\rho^{'AB})}+2(1-p)\Tr{(N^{B}_1\rho_B)},
\end{equation}
where $\mathcal{B}=M^{A}_1\otimes(N^{B}_1+N^{B}_2)+M^{A}_2\otimes(N^{B}_1-N^{B}_2)$ and $\mathcal{B}'=m^{A}_1\otimes(N^{B}_1+N^{B}_2)+m^{A}_2\otimes(N^{B}_1-N^{B}_2)$ with $m^{A}_{1(2)}$ being the $2\times2$ block matrix of $M^{A}_{1(2)}$. In~(\ref{twoTraces}) we have defined
\begin{equation}
\rho^{'AB}=\op{\tilde{\psi}}{\tilde{\psi}},\quad\rho^B=\Tr_A{\rho^{'AB}}=\begin{pmatrix}
\lambda & 0\\
0& 1-\lambda
\end{pmatrix}.
\end{equation}
\par
The maximum possible value of the first term in~\ref{twoTraces} given by the Horodecki criterion is $2\sqrt{1+4\lambda(1-\lambda)}$. This maximum can be achieved by the measurement settings $N_1=\sigma_z$, $N_2=\sigma_x$, $m_1= \cos{\theta}\sigma_z+\sin{\theta}\sigma_x$, $m_2= \cos{\theta}\sigma_z-\sin{\theta}\sigma_x$, where $\cos{\theta}=1/\sqrt{1+4\lambda(1-\lambda)}$. Notice that this choice of measurements also maximizes the second term with the value $2\lambda-1$. Hence, the CHSH violation is given by
\begin{equation}
2p\sqrt{1+4\lambda(1-\lambda)}+(1-p)(2\lambda-1)\le 2.
\end{equation}
One can check that when $p\leq 1/2$, the CHSH Inequality holds for all quantum states $\ket{\tilde{\psi}}$. Therefore, the erasure channel is CHSH-breaking for $p\leq 1/2$.
\section*{References}
\bibliography{references.bib}
\end{document}